\documentclass{llncs}

\title{Model Checking for Modal Dependence Logic:\\An Approach Through Post's Lattice}

\author{Julian-Steffen M\"uller \and Heribert Vollmer}

\institute{Institut f\"ur Theoretische~Informatik\\%
Leibniz Universit\"at Hannover\\%
Appelstr.~4, 30167~Hannover, Germany\\%
\email{\{mueller,vollmer\}@thi.uni-hannover.de}}

\usepackage{mathrsfs,amssymb, amsmath, stmaryrd}
\usepackage[usenames,dvipsnames]{xcolor}
\usepackage{graphicx}
\usepackage{ifthen}

\usepackage{tikz}
\usetikzlibrary{arrows,shapes}
\usepackage{algorithm2e}
\usepackage{subfig}

\tikzstyle{clone}=[x=1cm,y=.5cm,ellipse,draw=black,thick,minimum size=5.5mm,inner sep=1mm]
\tikzstyle{cloneRelation}=[->]
\tikzstyle{complexityClass1}=[fill=Goldenrod!60]
\tikzstyle{complexityClass2}=[fill=LimeGreen!80]
\tikzstyle{complexityClass3}=[fill=Red]

\newcommand\problemdef[4]%
{%
\begin{center}
\fbox{%
\begin{minipage}{0.91\textwidth}
\begin{tabular}{rp{9.1cm}}
\textbf{Problem:} & #1\\
\textbf{Description:} &#2.\\
\textbf{Input:} & #3.\\
\textbf{Question:}&#4?
\end{tabular}%
\end{minipage}%
}
\end{center}
}%

\spnewtheorem{observation}[note]{Observation}{\bfseries}{\itshape}

\newcommand{\defEQ}{:=}

\newcommand{\Model}{\mathscr{M}}
\newcommand{\Clone}[1]{\mathrm{#1}}
\newcommand{\CloneID}{\Clone{ID}}
\newcommand{\CloneE}{\Clone{E}}
\newcommand{\CloneV}{\Clone{V}}
\newcommand{\CloneM}{\Clone{M}}
\newcommand{\CloneN}{\Clone{N}}
\newcommand{\CloneL}{\Clone{L}}
\newcommand{\CloneBF}{\Clone{BF}}

\newcommand{\MC}[1]{#1\textnormal{-MC}}
\newcommand{\MDL}{\mathrm{MDL}}
\newcommand{\MDLMC}{\MC{\MDL}}

\newcommand{\NLOGSPACE}{\mathrm{NL}}
\newcommand{\CO}[1]{\textrm{co-}#1}

\newcommand{\PTIME}{\mathrm{P}}
\newcommand{\NPTIME}{\mathrm{NP}}
\newcommand{\PPOWERNP}[1][1]{\PTIME^{\NPTIME[#1]}}
\newcommand{\PSPACE}{\mathrm{PSPACE}}

\newcommand{\QBF}{\mathrm{QBF}}
\newcommand{\SAT}{\mathrm{SAT}}

\newcommand{\CNF}{\mathrm{CNF}}

\newcommand{\evalTerm}[3][\Model]{\langle{#2}\rangle_{#3}^#1}
\newcommand{\succTeamSet}[1]{R\langle#1\rangle}
\newcommand{\succTeam}[2][R]{#1(#2)}
\newcommand{\lTrue}{\mathrm{t}}
\newcommand{\lFalse}{\mathrm{f}}
\newcommand{\lAnd}{\wedge}
\newcommand{\lOr}{\vee}
\newcommand{\lSJ}{\otimes}
\newcommand{\lAtomicNeg}[1]{\overline{#1}}
\newcommand{\lBox}{\protect\ensuremath{\Box\kern-1.2ex\raisebox{.2ex}{$\cdot$}\kern0.55ex}}
\newcommand{\lDiamond}{\Diamond}
\newcommand{\lAll}{\Box}
\newcommand{\lDep}[1][\cdot]{{\textnormal{dep}\ifthenelse{\equal{#1}{}}{}{(\nobreak#1\nobreak)}}}
\newcommand{\lNeg}{\neg}
\newcommand{\lParity}{\oplus}

\newcommand{\REACH}{\mathrm{REACH}}
\begin{document}
	\maketitle
	
	\begin{abstract}
		In this paper we investigate an extended version of modal dependence logic by allowing arbitrary Boolean connectives. Modal dependence logic was recently introduced by Jouko V\"a\"an\"anen by extending modal logic by a the dependence atom $\lDep$. In this paper we study the computational complexity of the model checking problem. For a complete classification of arbitrary Boolean functions we are using a Lattice approach introduced by Emil Post. This classification is done for all fragments of the logical language allowing modalities $\lDiamond$ and $\lAll$, the dependence atom, and logical symbols for arbitrary Boolean functions.
		
	\end{abstract}

	\section{Introduction}
	
	Many algorithmic problems for propositional logic and its extensions are presumably computationally intractable, the most prominent of course the simple satisfiability problem SAT known to be NP-complete. For propositional modal logic, satisfiability is even PSPACE-complete \cite{la77}. Much effort has therefore been spent on identifying fragments of the logical language that admit efficient algorithms for satisfiability, see \cite{le79} for propositional logic and \cite{hescsc10} for modal logic. These studies first extend propositional (modal) logic by allowing arbitrary Boolean connectives (i.e., logical symbols for arbitrary Boolean functions) in the formulas, and then classify the computational complexity of satisfiability for each finite subset $B$ of allowed Boolean functions/connectives. An important tool in these complexity classifications is Post's lattice of all closed classes of Boolean functions, also known as Boolean clones, since it can be shown that the complexity of satisfiability for logic with connectives from $B$ depends only on the clone $[B]$ generated by $B$.

In this paper we are interested in modal dependence logic. This logic extends (propositional) modal logic by dependence atoms, i.e., atomic formulas that describe functional dependencies between variables. This logic was introduced recently by V\"a\"an\"anen \cite{va08} and examined from a complexity theoretic point of view in \cite{se09,eblo12,lovo10}.
While the model checking problem for propositional modal logic is known to be efficiently solvable (i.e., in polynomial time) \cite{filad79}, it gets PSPACE-complete for modal dependence logic. The above sketched approach to identify efficiently solvable fragments making use of the structure of Post's lattice does not work here, because the semantics of the Boolean connectives is not immediate in dependence logic. For example, $\lSJ$ (here called splitjunction) and $\rightarrow$ (intuitionistic implication) are defined in somewhat non-classical ways making use of so called team-semantics, see \cite{va08,se09}.
Ebbing and Lohmann \cite{eblo12} examined the complexity of a few fragments, but the fragments were given by somewhat arbitrary bases; their results determine the complexity of model checking in some important special cases, but the full picture is still missing.

In the present paper we introduce a novel approach to the study of fragments of dependence logic: We do not aim at a classification of all fragments defined by arbitrary dependence connectives like splitjunction or intuitionistic implication. Instead we make a distinction between dependence connectives on the one hand side and classical Boolean connectives on the other side. In other words, we introduce connectives given by Boolean function into dependence logic and define their semantics in the classical way. Then it can be observed that for this latter class of connectives an approach via Post's lattice is possible, and this is what we exploit in this paper. We achieve a classification of the model checking problem for modal dependence logic for all fragments of the language making use of dependence atoms, one or both modalities, and arbitrary Boolean connectives. 
As we will explain, the complexity will depend not on the particular choice of Boolean functions that we allow in our formulas, but on the clones in Post's lattice that is defined by the set of Boolean functions. In this way, the mentioned results from \cite{eblo12} will allow us more generally to determine the complexity of model checking for all monotone clones. We then extend these observations to all the remaining clones by considering also the connectives of logical negation and exclusive-or.

For the results presented here, we do not consider dependence connectives (splitjunction, intuitionistic implication, etc.), but we come back to this question in the conclusion.
	
After introducing the reader to dependence logic and our extension via arbitrary classical Boolean connectives we shortly recall basic results about Post's lattice in Sect.~\ref{prelim}. Then, in Sect.~\ref{results} we prove our classification results. We will see that when restricting the language to the modality $\lAll$, model checking becomes a very efficiently solvable task, independently of what else we allow in our language. Introducing the modality $\lDiamond$, however, makes model checking hard. We obtain fragments that are NP-complete, some are complete for $\mathrm P^{\mathrm{NP}[1]}$. The technically most interesting theorem of our paper shows that as soon as the connective exclusive-or is present or can be simulated, model checking reaches its maximal complexity and becomes PSPACE-complete.

	\section{The Modal Language and Its Fragments}\label{prelim}

We first define syntax and semantics of modal dependence logic.
	
	\begin{definition}\rm
		Let $B$ be a set of Boolean functions. Then we define the set of $\MDL_{B}$-formulae ($B$-formulae for short) as follows: 
		Every variable $p$ is a ${B}$-formula. 
		If $p_{1},\dots,p_{n},q$ are variables, then $\lDep[p_{1},\dots,p_{n},q]$ is a ${B}$-formula.
		If $f$ is an $n$-ary function in $B$ and $\phi_1,\dots,\phi_n$ are ${B}$-formulae, then $f(\phi_1,\dots,\phi_n)$ is a ${B}$-formula.
		If $\phi$ is a ${B}$-formula, then $\lDiamond \phi$ and $\lAll \phi$ are ${B}$-formulae.		
		
For $U\subseteq\{\lAll,\lDiamond,\lDep\}$ we say that a $B$-formula is a $(B,U)$-formula, if it uses only logical symbols from $B\cup U$.
	\end{definition}

	We remark that, as usual, we do not distinguish in our notation between a Boolean function $f$ and a logical symbol for $f$.
	
	The dependence atom $\lDep[p_{1},\dots,p_{n},q]$ is meant to express that the value of $q$ functionally depends on those of $p_{1},\dots,p_{n}$.
Unlike in usual modal logic, it does not make sense to evaluate such a formula in a single state but in a set of states 
(in this context called \emph{team}), and this is different from
evaluating the formula in each state separately.

As usual, in a Kripke structure $\Model= (W,R,\pi)$ the set of all successors of $T\subseteq W$ is defined as $\succTeam{T} = \{s \in W \mid \exists s' \in T:  (s',s) \in R\}$.
Furthermore we define $\succTeamSet{T} = \{T'\subseteq \succTeam{T} \mid \forall s\in T ~ \exists s' \in T' : (s,s') \in R\}$.

	\begin{definition}\rm
		Let $\Model$ be a Kripke structure, $T$ be a team over $\Model$ and $\phi$ be a ${B}$-formula. The the semantic evaluation (denoted as $\Model,T \models \phi$) is defined by the induction below. We also define the function 
$\evalTerm{\cdot}{T}$ which maps a formula to a truth value, 
where $\evalTerm{\phi}{T}$ is true if and only if $\Model,T \models \phi$.
		
		$$\begin{array}{lll}
			\Model, T \models p & \textnormal{ if } & \textnormal{for all } w \in T: p \in \pi(w)\\
			\Model, T \models \lAtomicNeg{p} & \textnormal{ if } & \textnormal{for all } w \in T: p \not\in \pi(w)\\
			\Model, T \models \lDep[p_{1},\dots,p_{n},q] & \textnormal{ if } & \textnormal{for all } w,w' \in T :\\
			&& \pi(w) \cap \{p_{1},\dots,p_{n}\} = \pi(w') \cap \{p_{1},\dots,p_{n}\}\\ 
			&& \textnormal{implies }q \in \pi(w) \Leftrightarrow q \in \pi(w')\\ 
			\Model, T \models f(\phi_1,\dots,\phi_n) & \textnormal{ if } & f(\evalTerm{\phi_1}{T},\dots,\evalTerm{\phi_n}{T})=1 \\
			\Model , T \models \lDiamond \phi & \textnormal{ if } & \textnormal{there is a } T' \in \succTeamSet{T} \textnormal{ such that } \evalTerm{\phi}{T}\\
			\Model , T \models \lAll \phi & \textnormal{ if } & \evalTerm{\phi}{\succTeam{T}}\\
		\end{array}$$
	\end{definition}

These modalities, as defined by V\"a\"an\"anen, do not fulfill the usual dualities; as a technical tool for our upcoming results we therefore define a further modality by $\lBox\phi \equiv \neg\lDiamond\neg\phi$. Also, note that $\lAll\phi \equiv \neg\lAll\neg\phi$.

We collect some important observations, all of which follows quite immediately from the definitions. 

	\begin{lemma}\label{lemma:lAll}
		Let $\phi,\phi'$ be $\MDL$ formulae. Then the following axioms are satisfied on all Kripke models $\Model$.
		\begin{enumerate}
			\item $\lAll(\phi \lAnd \phi') \rightarrow \lAll\phi \lAnd \lAll\phi'$.
			\item $\lAll(\phi \lOr \phi')  \rightarrow \lAll\phi \lOr \lAll\phi'$.
			\item $\lAll(\lNeg\phi)  \rightarrow \lNeg\lAll\phi$.
		\item 		Let $f^n$ a n-ary Boolean formula over the basis $\Clone{B}$. Then $\lAll{f(\phi_1,\dots,\phi_n)} \rightarrow f(\lAll\phi_1,\dots,\lAll,\phi_n)$ holds.
		\end{enumerate}	
	\end{lemma}
	\begin{proof}
		Let $\phi, \phi'$ be $\MDL$ formulae, $\Model$ be a Kripke model and $T$ be a team over $\Model$. Let $\odot \in \{\lOr, \lAnd\}$. Then the axioms 1 and 2 will follow by simple equivalencies from the definition as follows:
		$$\begin{tabular}{lcl}
				$\Model, T \models \lAll(\phi \odot \phi')$ & $\Leftrightarrow$ & $\Model, \succTeam{T} \models (\phi \odot \phi')$\\ 
				& $\Leftrightarrow$ & $\Model, \succTeam{T} \models \phi \text{ ''$\odot$''} \Model, \succTeam{T} \models \phi')$\\
				& $\Leftrightarrow$ & $\Model, T \models \lAll\phi \odot \lAll\phi'$
			\end{tabular}$$
		The axiom 3 is stating a self duality propertie of $\lAll$ and can be proven by the following equivilance.
		
		$$\begin{tabular}{lcl}
				$\Model, T \models \lAll\lNeg\phi$ & $\Leftrightarrow$ & $\Model, \succTeam{T} \models \lNeg\phi \hspace*{0.3mm} \Leftrightarrow$ $\Model, \succTeam{T} \not\models \phi$\\ 
				& $\Leftrightarrow$ & $ \Model, \succTeam{T} \not\models \lAll\phi \Leftrightarrow \Model, T \models \lNeg\lAll\phi$
			\end{tabular}$$
		Axiom 4 directly follows from the axioms 1-3, because each Boolean function $f$ can be efficiently transformed into a logically equivalent function $f'$ over the basis $\{\lAnd, \lOr, \lNeg\}$. By applying the axioms 1-3 iteratively on $f'$ we obtain the axiom 4.
		\qed
	\end{proof}
	
	The algorithmic problem family whose computational complexity we want to determine in this paper is defined as follows. Here, $B$ denotes a finite set of Boolean functions, and $U\subseteq\{\lAll,\lDiamond,\lDep\}$.
		
	\problemdef{$\MDLMC({B,U})$}{Model checking problem for $(B,U)$-formulae}{$(B,U)$-formula $\phi$, Kripke model $\Model$ and team $T$}{Is $\phi$ satisfied in $\Model$ on $T$}

%
%
%
%
	

	\begin{figure}[htp!]
		\centering
		\begin{tikzpicture}[scale=0.7]			
			\node[clone, label={[xshift=.7cm, yshift=-.5cm]$\emptyset$}] (ID) at (0,0)  {$\CloneID$};
			\node[clone, label={[xshift=.7cm, yshift=-.6cm]$\{\lAnd\}$}] (E) at (-3,4) {$\CloneE$};
			\node[clone, label={[xshift=.7cm, yshift=-.6cm]$\{\lOr\}$}] (V) at (-1,4) {$\CloneV$};
			\node[clone, label={[xshift=.9cm, yshift=-.6cm]$\{\lAnd, \lOr\}$}] (M) at (-2,8) {$\CloneM$};
			\node[clone, label={[xshift=.7cm, yshift=-.6cm]$\{\lNeg\}$}] (N) at (1,4) {$\CloneN$};
			\node[clone, label={[xshift=.7cm, yshift=-.6cm]$\{\oplus\}$}] (L) at (1,8) {$\CloneL$};
			\node[clone, label={[xshift=1cm, yshift=-.6cm]$\{\lAnd,\lNeg\}$}] (BF) at (-0.5,12) {$\CloneBF$};
			
			\path[cloneRelation] (ID) edge (E) (ID) edge (V) (ID) edge (N);
			\path[cloneRelation] (E) edge (M) (V) edge (M);
			\path[cloneRelation] (M) edge (BF);
			\path[cloneRelation] (N) edge (L);
			\path[cloneRelation] (L) edge (BF);
		\end{tikzpicture}
		\caption{Post's Lattice for Boolean clones with both constants and their ``standard'' basises}\label{mc-clones}
	\end{figure}
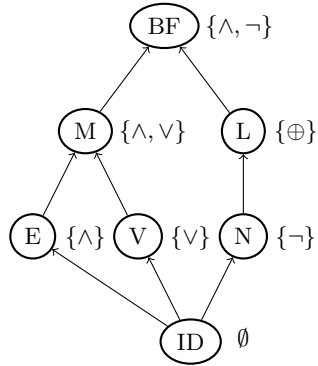

\subsection*{Post's Lattice}	
	  
  Emil Post \cite{pos41} classified the lattice of all closed sets of Boolean functions---called \emph{clones}---and 
  found a finite base for each clone. For an arbitrary finite set $B$ of Boolean functions we define $[B]$ to be the clone generated by $B$, i.e., the class of all Boolean functions that contains $B$, all projections (identities), and is closed under composition.
  A list of all clones as well as the full inclusion graph can be found, for example, in~\cite{bcrv03}.
  Whereas in general there is an infinite set of clones,
  for model checking luckily there are only seven different clones~\cite{bmsssh11}.
  This is essentially due to the fact that
  the constants for \emph{false} and \emph{true} do not need to be part of the language
  but can be expressed by atoms that are either nowhere or everywhere satisfied in the model. In other words, $\MDLMC(\Clone{B})\equiv\MDLMC(\Clone{B}\cup\{0,1\})$ ($\equiv$ denotes a suitable reductions, e.g., polynomial-time, logspace, or even constant-depth).
  
But there are only seven clones that contain the constants, see Fig.~\ref{mc-clones}. This means (and is proved formally in \cite{thomas12}) that if one wishes to study the computational complexity of model checking for propositional formulas with logical connectives restricted to some set $B$ of Boolean functions, it is not necessary to consider all infinite possibilities for such sets $B$ but actually suffices to consider these seven clones, depicted in Fig.~\ref{mc-clones}, where we describe the clones by their standard bases 
(we use $\oplus$ to denote the exclusive or).

As an example, notice that even though $\{\wedge,\oplus\}$ is not a base for all Boolean functions,
  it suffices to express all Boolean functions w.r.t.\ model checking problems
  because of the ``free'' existence of the constants; e.g., $\neg x=x\oplus1$ and $x\vee y=\left((x\oplus1)\wedge(y\oplus1)\right)\oplus1$.

To summarize, given any finite set $B$ of Boolean functions/propositional connectives, the computational complexity of model checking for formulas over $B$ is equivalent to the complexity of model checking for one of the bases given in Fig.~\ref{mc-clones}. 


Hence, in all upcoming results, if we classify the computational complexity of a model checking problem for the bases in Fig.~\ref{mc-clones}, we have in fact achieved a full complexity classification for all finite sets $B$ of Boolean connectives.

	%
	%
	
	%
	%
	
	\section{Complexity Results}\label{results}

	We first study fragments of the modal language with $\lAll$ as only modality. The following theorem completely clarifies the complexity of all arising fragments. 
	
	\begin{theorem}
For all finite sets $B$,
$\MC{\MDL}(B,\{\lAll,\lDep\})$ 
is $\NLOGSPACE$-com\-plete.
	\end{theorem}
	\begin{proof}
To prove hardness, we give a reduction from the standard NL-complete graph reachability problem.
		Let $\langle G=(V,E), s,t \rangle$ be a instance of $\REACH$, then we construct a Kripke model $\Model=(W,R,\pi)$ as follows:
			$$\begin{array}{cll}
				(W, R) & \defEQ &  (V, E\cup \{(v,v)\ |\ v \in V\})\\
				\pi(q) & \defEQ & \; V \setminus \{t\}\\
			\end{array}$$
Now we conclude: If there is no path in $G$ from $s$ to $t$, then $t$ is not contained in in one of the first $|V|-1$ breath depth first search levels. By the definition of $\pi$ it holds, that all vertices in the first $|V|-1$ levels are labeled with the proposition $q$ and therefore $\lAll^{(|W|-1)}\lDep[q]$ holds on $\Model$ at the starting team $\{s\}$. The converse direction is proved similarly. 

The membership result uses a well known fact by Buss \cite{bu87}, that propositional formulae can be evaluated in $\NLOGSPACE$. Because of Lemma~\ref{lemma:lAll} we know that modalities can only occur as a sequence at the leafs of the formula tree followed by an atomic formula. Every time the Buss algorithm needs to evaluate such a modal leaf we evaluate that leaf in $\NLOGSPACE$ with Algorithm \ref{alg:conl}  and the Buss algorithm can proceed with the corresponding Boolean value. This procedure is shown in the algorithm given in the appendix, where $\phi_\text{leaf}$ is such a modal leaf and $\text{Depth}_\text{modal}(\phi_\text{leaf})$ gives the length of the modal sequence.

\begin{algorithm}[htp!]
	\caption{$\CO{\NLOGSPACE}$ leaf checking algorithm $\MC{\MDL}(\CloneBF,\{\lAll\})$}
	\label{alg:conl}
	\SetKwInOut{Input}{Input}
	\SetKwInOut{Output}{Output}
	\SetKw{UniGuess}{universally guess}
	\SetKw{ExGuess}{existentially guess}
	\SetKw{KwNot}{not}
	\SetKw{KwAnd}{and}
	\SetKw{KwOr}{or}
	\SetKw{KwTrue}{true}
	\SetKw{KwFalse}{False}
	\SetKw{KwAccept}{accept}
	\SetKw{KwReject}{reject}

	\Input{$\MDL_\CloneBF$ formula $\phi$, Kripke model $\Model=(W,R,\pi)$, team $T \subseteq W$ and leaf $\phi_\text{leaf}$}
	\Output{Is $\phi_\text{leaf}$  satisfied in $\Model$ on $T$?}

	\UniGuess $w_1 \in W$ with $d(t,w_1) = \text{Depth}_{\text{modal}}(\phi_\text{leaf})$ for $t \in T$

	\UniGuess $w_2 \in W$ with $d(t,w_2) = \text{Depth}_{\text{modal}}(\phi_\text{leaf})$ for $t \in T$

	\uIf{$\phi_\text{leaf} = \lDep[p_1,\dots,p_n]$}{
		labellingAgrees $\leftarrow$ \KwTrue

		\For{$i \leftarrow 0$ \KwTo $n-1$} {
			\If{
				\KwNot $(p_i \in \pi(w_1) \Leftrightarrow p_i \in \pi(w_2))$}{
				labellingAgrees $\leftarrow$ \KwFalse
			}
		}

		\If{labellingAgrees} {
			\If {\KwNot $(p_n \in \pi(w_1) \Leftrightarrow p_n \in \pi(w_2))$}{
				\KwReject
			}
		}
	}
	\uElseIf{$\phi_\text{leaf} = p$} {
		\If{$p \not\in \pi(w_1)$}{
			\KwReject
		}
	}
	\ElseIf{$\phi_\text{leaf} = \lAtomicNeg{p}$} {
		\If{$p \in \pi(w_1)$}{
			\KwReject
		}
	}

	\KwAccept
\end{algorithm}

\qed
\end{proof}


	%
	%
	
	In the rest of the section we study the modal language with modality $\lDiamond$. However, we will see that all obtained classifications hold as well for the modal language with both modalities $\lDiamond$ and $\lAll$. The results we will obtain are summarized in Fig.~\ref{diamond}.
	

	\begin{figure}[htp!]
		\centering
		\begin{tikzpicture}[scale=0.7]
			\node[clone,complexityClass1] (ID) at (4,0)  {$\CloneID$};
			\node[clone,complexityClass1] (E) at (1,4) {$\CloneE$};
			\node[clone,complexityClass1] (V) at (3,4) {$\CloneV$};
			\node[clone,complexityClass1] (M) at (2,8) {$\CloneM$};
			\node[clone,complexityClass2] (N) at (5,4) {$\CloneN$};
			\node[clone,complexityClass3] (L) at (5,8) {$\CloneL$};
			\node[clone,complexityClass3] (BF) at (3.5,12) {$\CloneBF$};
			
			\path[cloneRelation] (ID) edge (E) (ID) edge (V) (ID) edge (N);
			\path[cloneRelation] (E) edge (M) (V) edge (M);
			\path[cloneRelation] (M) edge (BF);
			\path[cloneRelation] (N) edge (L);
			\path[cloneRelation] (L) edge (BF);
			
			\node[clone,complexityClass1,label=right:{$\NPTIME$}] (class1) at (8,1) {};
			\node[clone,complexityClass2,label=right:{$\PPOWERNP$}] (class2) at (8,3) {};
			\node[clone,complexityClass3,label=right:{$\PSPACE$}] (class3) at (8,5) {};
		\end{tikzpicture}
		\caption{Complexity of $\MC{\MDL}(\{\lDiamond,\lDep\})$ and
				$\MC{\MDL}(\{\lDiamond,\lAll,\lDep\})$
				}\label{diamond}
	\end{figure}

	In order to prove the upper bound for the negation clone $\CloneN$, we have to prove a property called downwards closure for the universal modal operator.
	
			We say that a logic has the \emph{downwards closure property} if for all Kripke models $\Model$, all teams $T$ over $\Model$ and all formulae $\phi$, if $\phi$ is satisfied in $\Model$ on T, then it is satisfied on any subset $T'$ of $T$. We also say that in this case, $\phi$ is downards closed.

	\begin{lemma}\label{lemma_dc-box-equiv-all}
		Let $\phi$ be a downwards closed $\MDL$ formula. Then the formula $\lBox\phi$ is logically equivalent to $\lAll\phi$.
	\end{lemma}
	\begin{proof}
		Let $\phi$ be a downwards closed $\MDL$ formula, $\Model$ be a Kripke model, $T$ be team over $\Model$ and $\lBox\phi$ be satisfied in $\Model$ on $T$. By definition $\phi$ has to be satisfied on all successor teams in $\succTeamSet{T}$, especially on $\succTeam{T}$. Clearly all teams $T'$ in $\succTeamSet{T}$ are subsets of the team of all successors $\succTeam{T}$. Because of this and the fact that  $\phi$ is downwards closed, it is sufficient to check if $\phi$ evaluates to true on $\succTeam{T}$.
		\qed
	\end{proof}
	%
	%

	\begin{theorem}\label{N}
		Let $[B]=\CloneN$, then $\MC{\MDL}(B,\{\Diamond, \lDep\})$ is $\PPOWERNP$-complete.	
	\end{theorem}
	\begin{proof}
		Let $\phi$ be a $(B,\{\Diamond, \lDep\})$-formula and 
		$k_1,\dots,k_n \in \mathbb{N}$. Then $\phi$
		is always either of the form $\phi \defEQ \lNeg\phi'$ or $\phi \defEQ \phi'$, where
			$$\phi' \defEQ \Diamond^{k_1}\lBox^{k_2}\dots\lBox^{k_n}\lambda,\quad
				 \text{$\lambda$ is a literal or a dependence atom.}$$
		Let $\hat{\phi}$ be any sub formula of $\phi'$. Then it follows from Lemma \ref{lemma_dc-box-equiv-all}, that $\lBox\hat{\phi}$ can be  replaced by $\lAll\hat{\phi}$. Hence we can rewrite $\phi'$ as
			$$\phi' \defEQ \ \Diamond^{k_1}\lAll^{k_2}\dots\lAll^{k_n}\lambda,\quad  
				 \text{$\lambda$ is a literal or a dependence atom.}$$			
Thus, model checking for $\phi$ can be reduced clearly to model checking for $\phi'$.	Since Ebbing and Lohmann showed in \cite{eblo12} that MDL model checking for the operator fragment $\{\lDiamond,\lAll,\lDep,\lAtomicNeg{\cdot}\}$ is NP-complete, we conclude $\MC{\MDL}(\CloneN, \{\lDiamond, \lDep\})$ is in $\PPOWERNP$.
		
		It remains to show that $\MC{\MDL}(\{\CloneN, \lDiamond, \lDep\})$  is $\PPOWERNP$-hard.
		Let $A \in \PPOWERNP$, and let the corresponding Turing-Machine be $M_A$. We have to show that $A \leq^p_m$ $\MC{\MDL}(\CloneN, \{\lDiamond, \lDep\})$.
		
		In the polynomial many-one reduction, we can simulate the polynomial part of the machine. Therefore the only thing that is
		left, is the oracle question and four possible acceptance behaviours of $M_A$ as shown
		in Figure \ref{fig:THmcDTD-Modus}.
		
		\begin{figure}[h!]
			\centering
			\subfloat[$f \in \SAT$]{
				\label{fig:THmcDTD-ModusNP}
				\begin{tikzpicture}
					\draw (0,0) node (input) {input x};
					\draw (0,-1) node (oracle) {$f \in \SAT$};
					\draw (-1,-2) node (acc) {acc};
					\draw (1,-2) node (rej) {rej};

					\draw (input) -- (oracle);
					\draw (oracle) -- node[left=0.15cm] {$1$} (acc);
					\draw (oracle) -- node[right=0.15cm] {$0$}(rej);
				\end{tikzpicture}
			}
			\subfloat[$f \in \overline{\SAT}$]{
				\label{fig:THmcDTD-ModusCoNP}
				\begin{tikzpicture}
					\draw (0,0) node (input) {input x};
					\draw (0,-1) node (oracle) {$f \in \SAT$};
					\draw (1,-2) node (acc) {acc};
					\draw (-1,-2) node (rej) {rej};

					\draw (input) -- (oracle);
					\draw (oracle) -- node[left=0.15cm] {$1$} (rej);
					\draw (oracle) -- node[right=0.15cm] {$0$}(acc);
				\end{tikzpicture}
			}
			\subfloat[Accept always]{
				\label{fig:THmcDTD-ModusAccept}
				\begin{tikzpicture}
					\draw (0,0) node (input) {input x};
					\draw (0,-1) node (oracle) {$f \in \SAT$};
					\draw (1,-2) node (acc1) {acc};
					\draw (-1,-2) node (acc2) {acc};

					\draw (input) -- (oracle);
					\draw (oracle) -- node[left=0.15cm] {$1$} (acc2);
					\draw (oracle) -- node[right=0.15cm] {$0$}(acc1);
				\end{tikzpicture}
			}
			\subfloat[Reject always]{
				\label{fig:THmcDTD-ModusReject}
				\begin{tikzpicture}
					\draw (0,0) node (input) {input x};
					\draw (0,-1) node (oracle) {$f \in \SAT$};
					\draw (1,-2) node (rej1) {rej};
					\draw (-1,-2) node (rej2) {rej};

					\draw (input) -- (oracle);
					\draw (oracle) -- node[left=0.15cm] {$1$} (rej2);
					\draw (oracle) -- node[right=0.15cm] {$0$}(rej1);
				\end{tikzpicture}
			}
			\caption{Acceptance cases in $\PPOWERNP$.}
			\label{fig:THmcDTD-Modus}
		\end{figure}
		
		$\SAT$ is represented in $\MC{\MDL}(\CloneN, \{\lDiamond, \lDep\})$ in the same way it is represented by Ebbing and Lohmann in \cite{eblo12}, but we have to adjust our formula to represent the four possible acceptance cases.\\
		Let $\psi := \bigwedge_i C_i$ be the $\SAT$ oracle question and $g(\psi) = \langle\Model,T,\phi\rangle$ the reduction function.
	
		\noindent
		The Kripke structure $\Model=(W,R,\pi)$ is defined as follows: 
			$$W \defEQ \{c_1,\dots,c_n,s_1,\dots,s_m,\lAtomicNeg{s}_1,\dots,\lAtomicNeg{s}_m\},$$
		\begin{equation*}
			R \supseteq 
			\begin{cases}
				\{(c_i,s_j)\} & \text{, if } x_j \text{ occurs in } C_i\\
				\{(c_i,\lAtomicNeg{s}_j)\}& \text{, if } \lAtomicNeg{x}_j 
				\text{ occurs in } C_i,
			\end{cases}
		\end{equation*}
			$$\pi(s_i) \supseteq \{p_i,q\},\quad\quad \pi(\lAtomicNeg{s}_i) \supseteq \{p_i\}.$$
		The initial team is defined by the worlds, which representing the clauses of $\psi$.
			$$T \defEQ \{c_1,\dots,c_n\},$$
					
		\begin{equation*}
			\phi \defEQ 
				\begin{cases}
					\Diamond\lDep[p_1,\dots,p_m,q]&\text{, } \psi \in \SAT\\
					\lNeg\Diamond\lDep[p_1,\dots,p_m,q]&\text{, } \psi \in \overline{\SAT}\\
					\lTrue&\text{, accept always}\\
					\lFalse&\text{, reject always}.\\
				\end{cases}
		\end{equation*}

		The correctness follows directly from the $\MC{\MDL}(\CloneN, \{\lDiamond, \lDep\})$ correctness proof in
		\cite{eblo12} and the definition of $\lNeg$.
		\qed
	\end{proof}

	%
	%
	\begin{theorem}\label{L}
		Let $[B]\supseteq\CloneL$, then $\MC{\MDL}(B, \{\Diamond, \lDep\})$ is $\PSPACE$-complete.
	\end{theorem}
	
	\begin{proof}
To prove hardness, we give a reduction from the standard PSPACE-complete problem QBF, the validity problem for quantified Boolean formulae, to $\MC{\MDL}(\{\oplus\}, \{\Diamond, \lDep\})$.

		Let $\phi \defEQ \exists x_1 \forall x_2 \dots \exists x_n (C_i)$ be a QBF formula. The $\QBF$ formula will be transformed to a $\MC{\MDL}$ instance $\langle \Model=(W,R,\pi),T,\psi\rangle$ as follows. For each quantified variable we will construct one connected component. In these connected component the nesting of the variables will be simulated by \emph{value states} $x_i,\lAtomicNeg{x}_i$ and \emph{delay states} $d_i$. There are also components for each clause.	
		\begin{align*}
			W \defEQ	& \bigcup_{i=1}^n\big(\{x_i, \lAtomicNeg{x}_i\} \cup \{x_i^j, \lAtomicNeg{x}_i^j | i \leq j \leq n\}\ \cup \{d_i^j | 1 \leq j \leq i \}\big)\\
					& \ \cup\ \{c_i^j | 1 \leq i \leq m,\ 1 \leq j \leq n+1\} \\
		\end{align*}	
		For the quantified variable $x_i$ the variables value decision will be made at point $i$. At the decision point the natural ordering of delay states will branch in a natural ordering of the different values states.	
		\begin{align*}
			R \defEQ 	 \bigcup_{i=1}^n\bigg(&\{(x_i^j, x_i^{j+1}) | 1 \leq j < n-i\}\ \cup 
				  	  \{(\lAtomicNeg{x}_i^j, \lAtomicNeg{x}_i^{j+1}) | i \leq j \leq n\}\ \cup \\
					& \{(x_i^j, x_i) | j = n-i\}\ \hspace{0.22cm} \cup\ \{(\lAtomicNeg{x}_i^j, \lAtomicNeg{x}_i) |  j = n-i\}\ \cup \\
					& \{(d_i^j,d_i^{j+1}) |  1 \leq j < i \} \cup\  \{(d_i^i,x_i^{i+1}), (d_i^i, \lAtomicNeg{x}_i^{i+1})\} \cup\\
					& \{(c_i^j, c_i^{j+1}) | 1 \leq i \leq n\}\bigg) \cup \{(c_i,x_i | 1 \leq i \leq m,\ x_i \in C_i\}\\
		\end{align*}
		The starting team is the set of all initial delay nodes and the initial clause nodes.
		$$T \defEQ \{d_i^0 | 1 \leq i \leq n\} \cup \{c_i^1 | 1 \leq i \leq m\}$$	
		At last we have to define the labelling of the Kripke structure. For each positive value world $x_i^j,\lAtomicNeg{x}_i^j,\lAtomicNeg{x}_i,x_i$ $p_i$ is labelled to represent the variable and to represent the value on each positive value node $q$ is labelled also. 
		\begin{align*}
			\pi(p_i) \defEQ & \{x_i^j,x_i, \lAtomicNeg{x}_i^j,\lAtomicNeg{x}_i | 1 \leq i \leq n, i \leq j \leq n\} \\
			\pi(q) \defEQ & \{x_i^j,x_i | 1 \leq i \leq n, i \leq j \leq n\}
		\end{align*}
		The following formula $\psi$ will simulate the $\QBF$ evaluation on the given Kripke model $\Model=(W,R,\pi)$ over the team starting team $T$.
		\begin{align*}
			\psi \defEQ & \underbrace{\lDiamond\lBox\dots\lDiamond}_{\text{n-times}}\underbrace{(\lNeg\lDep[p_2,p_4,\dots,p_{n-1},q] \lParity \Diamond\lDep[p_1,\dots,p_n,q])}_{\psi'}\\
		\end{align*}
		
		\begin{figure}
			\centering
				\begin{tikzpicture}
					\tikzstyle{every node}=[circle,draw,text=black, minimum size=5mm, inner sep=0pt, font=\scriptsize]
					\tikzstyle{every path}=[->]
					\node  (x1) at (0,0) {$d_1^1$};
					\node  (x21) at (1,0.5)  {$x_2^2$};
					\node  (x20) at (1,-0.5) {$\lAtomicNeg{x}_3^2$};
					\node  (x31) at (2,0.5)  {$x_2^3$};
					\node  (x30) at (2,-0.5) {$\lAtomicNeg{x}_3^3$};
					\node  (x41) at (3,0.5)  {$x_1^4$};
					\node  (x40) at (3,-0.5) {$\lAtomicNeg{x}_3^4$};

					\node  (y1) at (0,-2) {$d_2^1$};
					\node  (y2) at (1,-2) {$d_2^2$};
					\node  (y31) at (2,-1.5) {$x_2^3$};
					\node  (y30) at (2,-2.5) {$\lAtomicNeg{x}_2^3$};
					\node  (y41) at (3,-1.5) {$x_2^4$};
					\node  (y40) at (3,-2.5) {$\lAtomicNeg{x}_2^4$};
					
					\node  (z1) at (0,-4) {$d_3^1$};
					\node  (z2) at (1,-4) {$d_3^2$};
					\node  (z3) at (2,-4) {$d_3^3$};
					\node  (z41) at (3,-3.5) {$x_3^4$};
					\node  (z40) at (3,-4.5) {$\lAtomicNeg{x}_3^4$};

					\node  (xv1) at (4.5,0.5)  {$x$};
					\node  (xv0) at (4.5,-0.5) {$\lAtomicNeg{x}$};
					\node  (yv1) at (4.5,-1.5) {$y$};
					\node  (yv0) at (4.5,-2.5) {$\lAtomicNeg{y}$};
					\node  (zv1) at (4.5,-3.5) {$z$};
					\node  (zv0) at (4.5,-4.5) {$\lAtomicNeg{z}$};

					\node  (c4) at (7,-1) {$c_1^4$};
					\node  (c3) at (8,-1) {$c_1^3$};
					\node  (c2) at (9,-1) {$c_1^2$};
					\node  (c1) at (10,-1) {$c_1^1$};
					
					\node  (d4) at (7,-3) {$c_2^4$};
					\node  (d3) at (8,-3) {$c_2^3$};
					\node  (d2) at (9,-3) {$c_2^2$};
					\node  (d1) at (10,-3) {$c_2^1$};

					\path	(x1) edge (x21) edge (x20)
							(x21) edge (x31)
							(x20) edge (x30)
							(x31) edge (x41)
							(x30) edge (x40)
							(x41) edge (xv1)
							(x40) edge (xv0);

					\path	(y1) edge (y2)
							(y2) edge (y31) edge (y30)
							(y31) edge (y41)
							(y30) edge (y40)
							(y41) edge (yv1)
							(y40) edge (yv0);
							
					\path	(z1) edge (z2)
							(z2) edge (z3)
							(z3) edge (z41) edge (z40)
							(z41) edge (zv1)
							(z40) edge (zv0);

					%
					\path	(c1) edge (c2)
							(c2) edge (c3)
							(c3) edge (c4)
							(c4) edge[out=180, in=0] (xv1) edge[out=180, in=0] (yv0) edge[out=180, in=0] (zv0);	

					\path	(d1) edge (d2)
							(d2) edge (d3)
							(d3) edge (d4)
							(d4) edge[out=180, in=0, dashed] (yv1) edge[out=180, in=0, dashed] (xv1) edge[out=180, in=0, dashed] (zv1);
				\end{tikzpicture}
			\caption{Kripke Model for $\QBF$ formula $\phi \defEQ \exists{x}\forall{y}\exists{z} (x \lOr \lAtomicNeg{y} \lOr \lAtomicNeg{z}) \lAnd (x \lOr y \lor z)$}
		\end{figure}
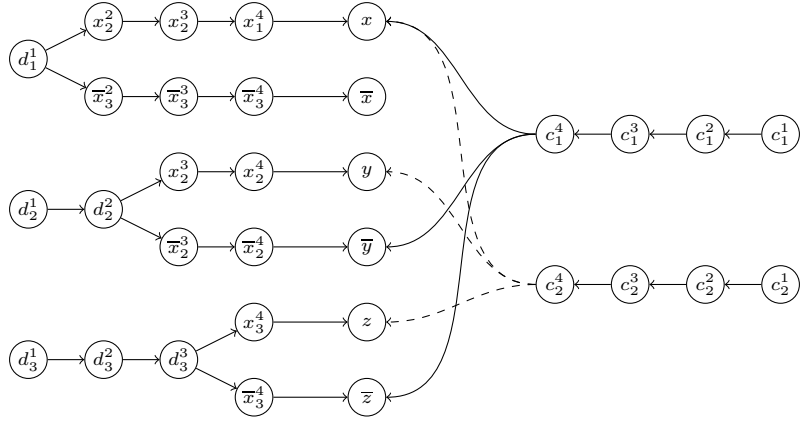
		
		Let $\phi'$ be a $\CNF$ over the variables $x_1,\dots,x_n$ and $\phi=\exists{x_1}\forall{x_2}\dots\exists{x_n}\phi'$ be a satisfied $\QBF$ formula. Then $\Model, T \models \psi$. After $n$ modal quantification steps, the formula $\psi'$ will be evaluated on teams $T'$ over $\{x_1,\lAtomicNeg{x}_1, \dots, x_n,\lAtomicNeg{x}_n\}$, because in the $i$'s modal step we pick one ore both variable vertices in the connected component corresponding to the $i$'s variable. For convenience we say that a team is consistent if is does not contain a variable positively and negatively and not consistent if it does. In the following we want to choose satisfying teams with respect to the $\QBF$ assignment tree and show that these teams satisfy $\psi'$. In the case of existential quantification we can choose the variable path with respect to the $\QBF$ assignment, but for the universal quantification we have to ensure that the case of both variable assignments are picked, does not falsify $\psi'$.
		
		\begin{claim}[1]\label{pspace_claim1}
			Let $T'$ be a team over $\{x_1,\lAtomicNeg{x}_1, \dots, x_n,\lAtomicNeg{x}_n\}$, where the the universal quantified variable $x_i$ is contained positively and negatively. Then $\Model, T'\models \psi'$ holds.
		\end{claim}
		
To prove the claim, let $x_i$ and $\lAtomicNeg{x}_i$ be in T'. Then $\neg\lDep[p_2,p_4,\dots,p_{n-1},q]$ is true, because T' does not satisfy the dependence atom. By this it follows that $\Diamond\lDep[\vec{p},q]$ will also be false, because $\{p_2,p_4,\dots,p_{n-1}\} \subseteq \{p_1,\dots,p_n\}$ and the modal operator does not shrink the team $T'$. This proves the claim.
		
		With this claim in mind, we construct all consistent successor teams of the form $T'=\bigcup_{i=1}^{n}\{t_i^n\} \cup \bigcup_{i=1}^{m}\{c_i^k\}$, which are representing the assignments of a $\QBF$ assignment tree. In the last claim we show that a consistent team evaluated to true on $\psi'$ if and only if the corresponding variable assignments satisfies $\phi'$. 
		
Hardness now follows from the following claim.
		\begin{claim}[2]\label{pspace_claim2}
			Let T be of the form $\bigcup_{i=1}^{n}\{t_i^n\} \cup \bigcup_{i=1}^{m}\{c_i^k\}$, where $t_i \in \{x_i^n, \lAtomicNeg{x}_i^n\}$ and $\alpha(x_i) = 
			\begin{cases}
				1 & , t_i = \{x_i^n\} \\ 
				0 & \text{, otherwise}
			\end{cases}$. Then $\Model, T \models \psi'$ if and only if $\alpha \models \phi'$.
		\end{claim}
		We prove the claim. Because of team $T$ is being consistent the sub formula $\neg\lDep[p_2,p_4,\dots,p_{n-1},q]$ is false. It remains to show that $\Diamond\lDep[p_1,\dots,p_n,q]$ is true if and only if $\alpha \models \phi'$ holds and false otherwise. Let $\alpha \models \phi'$. Then each clause $C_j$ is satisfied by at least one variable assignment in $\alpha$. W.lo.g. let $C_j$ be satisfied by $x_i=0$. By definition it follows that $C_j$ is connected with the vertex $\lAtomicNeg{x}_i$ and by $\alpha(x_i)=0$ that $\lAtomicNeg{x}_i^n \in T$. Then by consistency of $T$ and $\alpha$ it directly follows that a consistent successor team $T'$ with $\lAtomicNeg{x}_i \in \succTeamSet{T}$ and $x_i \not\in \succTeamSet{T}$ exists. 

Finally, $\MC{\MDL}({B, \{\Diamond,\lDep\}}) \in \PSPACE$ follows from Algorithm~\ref{alg:pspace} evaluating the formula in the obvious way.

\begin{algorithm}[htp!]
	\caption{$\PSPACE$ algorithm $\textbf{check}(\Model,T,\phi)$}
	\label{alg:pspace}
	\SetKwInOut{Input}{Input}
	\SetKwInOut{Output}{Output}
	\SetKw{UniGuess}{universally guess}
	\SetKw{ExGuess}{existentially guess}
	\SetKw{KwNot}{not}
	\SetKw{KwAnd}{and}
	\SetKw{KwOr}{or}
	\SetKw{KwTrue}{true}
	\SetKw{KwFalse}{False}
	\SetKw{KwAccept}{accept}
	\SetKw{KwReject}{reject}
	\SetKw{KwReturn}{return}
	\SetKw{KwCheck}{check}

	\Input{$\MDL_\CloneBF$ formula $\phi$, Kripke model $\Model=(W,R,\pi)$ and team $T \subseteq W$}
	\Output{Is $\phi$ is satisfied in $\Model$ on $T$?}

	\uIf{$\phi = f(\phi_1, \dots, \phi_n)$)}{
		\KwReturn f(\KwCheck$(\Model, T, \phi_1), \dots, \KwCheck(\Model, T, \phi_n)$)
	}	
	\uElseIf{$\phi = p$} {
		\ForEach{$w_i \in T$}{
			\lIf{$p \not\in \pi(w_i)$}{
				\KwReturn false
			}
		}
		\KwReturn true

	}	
	\uElseIf{$\phi = \lAtomicNeg{p}$} {
		\ForEach{$w_i \in T$}{
			\lIf{$p \in \pi(w_i)$}{
				\KwReturn false
			}
		}
		\KwReturn true

	}	
	\uElseIf{$\phi = \lDep[p_1,\dots,p_n,q]$} {
		\ForAll{$p_i \textbf{ in } \{ p_1, \dots, p_n \}$}{
			\ForAll{$w_{t_j} \textbf{ in } \{ w_{t_1}, \dots, w_{t_m} \} =T$}{
				$c_{ij} = \textnormal{valid}(\Model,\{ w_{t_j} \},p_i)$
			}
		}
			
		\For{$w_{t_i} \in T$}{
			\For{$w_{t_j} \in T$} {
				\If{$c_i = c_j$}{
					\If{$\textnormal{valid}(\Model,\{ w_{t_i} \},q) \not= \textnormal{valid}(\Model,\{ w_{t_j} \},q)$}{
						\KwReturn false
					}
				}
			}
		}
		\KwReturn true
	}	
	\uElseIf{$\phi = \lDiamond\phi'$} {
		guess existentially $ T' \in \succTeamSet{T}$

		\KwReturn \KwCheck$(\Model$,$T'$,$\phi')$
	}
	\ElseIf{$\phi = \lAll\phi'$} {
		\KwReturn \KwCheck($\Model, \succTeam{T}, \phi'$)
	}	
\end{algorithm}

\qed	
\end{proof}

From results presented by Ebbing and Lohmann in \cite{eblo12}, in which the NP-completeness of some particular fragments for modal dependence logic model checking was shown, it follows that if a set $B$ of Boolean connectives forms a base for one
of the Boolean clones $\CloneID, \CloneE, \CloneV, \CloneM$, it yields an NP-complete model checking problem.

Together with Theorem~\ref{N} and \ref{L} above, we thus obtain a complete picture for the complexity of model checking for modal dependence logic, as given in Fig.~\ref{diamond}. For any set $B$ of Boolean connectives, the complexity falls in one of the cases given there.
	
	\section{Conclusion}
	
	We obtained a complete classification of the complexity of the model checking problem for modal dependence logic for formulas that may contain dependence atoms, one or two modalities, and symbols for arbitrary Boolean functions.
	
	What we did not address here is formulas that besides the arbitrary sets of Boolean connectives also involve dependence connectives, e.g., splitjunction, intuitionistic implication, linear implication, etc.
While partial results, at least for the first two mentioned connectives, are known, a full classification is still missing. For the case of splitjunction combined with the diamond modality it is known that the classification shown in Figure~\ref{diamond} is still valid expect for the $\CloneN$ case. In this case no result is known. The classification of splitjucntion combined with the $\lAll$ modality misses the cases $\CloneL$ and $\CloneN$.

	Even more interesting is maybe the question how to develop a general concept of what a dependence connective is, and then study complexity issues concerning formulas with arbitrary sets of dependence connectives, maybe via a similar lattice as Post's. First steps into the direction of a concept of such general connectives have been made by Antti Kuusisto (personal communication).

\newpage		
			
	\bibliographystyle{splncs03}
	\bibliography{mdlpost}
\end{document}